\documentclass[aps,prl,twocolumn,superscriptaddress,nofootinbib,floatfix]{revtex4-2}
\usepackage[utf8]{inputenc}
\usepackage{amsmath}
\usepackage{amssymb}
\usepackage{amsthm}
\usepackage{graphicx}
\usepackage{subcaption}
\DeclareCaptionJustification{prl}{\leftskip=0pt\rightskip=0pt\parfillskip=0pt plus 1fil\relax}
\usepackage{bm}
\usepackage{physics}
\usepackage[x11names,table]{xcolor}
\usepackage[
  colorlinks=true,
  citecolor=blue,
  urlcolor=blue,
  linkcolor=blue,
  hypertexnames=false
]{hyperref}

\newtheorem{theorem}{Theorem}
\newtheorem{lemma}{Lemma}

\newtheorem{corollary}{Corollary}

\newcommand{\eps}{\varepsilon}
\newcommand{\sgn}{\operatorname{sgn}}

\newcommand{\supp}{\operatorname{supp}}
\DeclareMathOperator*{\argmin}{arg\,min}
\begin{document}

\title{Geometric View of Iterative Fixed-Node Dynamics}

\author{Pranav Kairon}
\email{pkairon2@illinois.edu}

\author{Bryan K. Clark}
\affiliation{The Anthony J. Leggett Institute for Condensed Matter Theory and IQUIST and
NCSA Center for Artificial Intelligence Innovation and Department of Physics,
University of Illinois at Urbana-Champaign, Illinois 61801, USA}

\date{\today}

\begin{abstract}
The sign-structure of a quantum many-body ground state is a central quantity in fixed-node approaches to mitigating the Fermion sign problem. For continuum electronic-structure Hamiltonians, the correct sign-structure is sufficient to compute exact ground state properties;  on the other hand lattice fixed-node approaches  retain an additional dependence on trial wave-function amplitudes which are improved upon non-optimally while preserving their sign structure.  Here we consider an iterative map obtained by repeatedly replacing the fixed-node trial state with the ground state of its associated lattice fixed-node Hamiltonian (independent of whether and
how a practical algorithm could implement this iteration).  We show the fixed points of this map are eigenstates of the Hamiltonian and reduced Hamiltonian resulting in at most one fixed point in the interior of each sign chamber.  We prove that the ground state fixed point is stable and that the boundary of the ground state sign chamber is repulsive.  This shows the  amplitude dependence beyond having the correct ground state sign-structure disappears under self-consistent iteration.  In other sign chambers, energy descent can be obstructed by the imposed sign structure; this tension leads to the iteration driving select amplitudes to zero producing `support collapse' onto chamber boundaries and corresponding to reduced Hamiltonians. Excited-state fixed points are therefore provably unstable and flow toward lower-energy boundary fixed points. We show sign chamber boundaries have a directional stability; boundary points attractive in one direction are repulsive under a sign flip.  By revealing the geometry and stability structure of iterative fixed-node dynamics, our results clarify the foundations of a central approach to the Fermion sign problem and motivate potential new algorithmic strategies.   

\end{abstract}

\maketitle

A nontrivial sign structure is one of the central features distinguishing quantum many-body states from classical probability distributions \cite{PhysRevLett.128.040403}.   
On a system of $n$ electrons, there are at least a doubly exponential  $2^{2^n}$  number of possible sign structures;  the fermion sign problem (FSP), which is the primary obstacle to simulating  quantum many-body physics, is at least partially caused by ignorance about the sector in which the ground state lives \cite{westerhout2023many,signproblemreview,signprob_majorana,signprob_science}.  For the continuum electronic structure Hamiltonian choosing the right sign structure completely removes the FSP;  fixed-node diffusion Monte Carlo (FNDMC) samples the best  state consistent with the sign structure of a trial wave-function $\Psi_T$ \cite{introtofndmc,fndmcmath1,fndmcmath2}. Interestingly, this is not the case for lattice problems clashing with the intuition that knowledge of the sign-structure is sufficient to solve the FSP. The analogous lattice fixed-node approach also samples a state with the same sign structure as an input trial wave-function $\Psi_T$ \cite{ten1995proof,lrdmc,lrdmc2}. However it is only a variationally improved (or equal) state, not the best state with that sign structure.\\ 

Lattice fixed-node methods have enabled studies of frustrated quantum magnetism, such as $J_1$-$J_2$ antiferromagnets~\cite{boninsegni1995triangular,sorella1998green,GMFCjj}, with stochastic reconfiguration providing a route to iteratively refine the fixed-node dynamics~\cite{sorella2000green}. They have also been used to investigate phase separation and insulating phases in Hubbard models~\cite{cosentini1998phase,tocchio2008backflow} and to benchmark ground-state energies~\cite{leblanc2015solutions}. More recently, variationally optimized neural-network backflow wave functions have served as trial wavefunctions for fixed-node diffusion Monte Carlo in the doped Hubbard model~\cite{loehr2025enhancing}. Recent developments have extended their relevance to quantum algorithms and complexity theory. Bravyi et al. used the construction to obtain a rapidly mixing continuous-time Markov chain that samples the ground-state probability distribution of a non-stoquastic Hamiltonian~\cite{bravyi2023rapidly}. Related constructions also design verification protocols for the local Hamiltonian problem with succinctly represented ground states~\cite{jiang2023local,waite2025complexitysuccinctstatelocal}.
\begin{figure*}[tbp]
  \centering
  \begin{subfigure}[t]{0.25\textwidth}
    \caption{}\label{fig:geometry-update}
    \includegraphics[width=\linewidth]{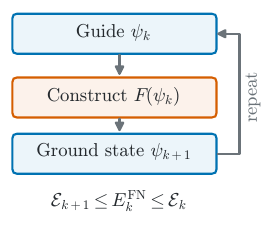}
  \end{subfigure}%
  \begin{subfigure}[t]{0.25\textwidth}
    \caption{}\label{fig:geometry-ground}
    \includegraphics[width=\linewidth]{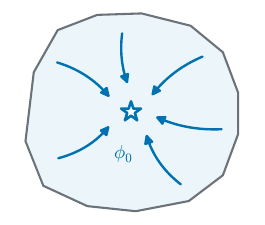}
  \end{subfigure}%
  \begin{subfigure}[t]{0.25\textwidth}
    \caption{}\label{fig:geometry-excited}
    \includegraphics[width=\linewidth]{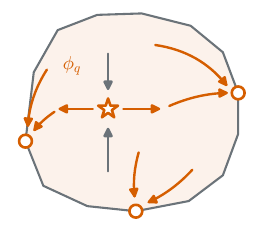}
  \end{subfigure}%
  \begin{subfigure}[t]{0.25\textwidth}
    \caption{}\label{fig:geometry-boundary}
    \includegraphics[width=\linewidth]{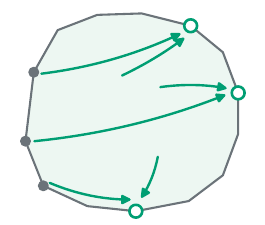}
  \end{subfigure}
  \caption{Schematic overview of iterative fixed-node dynamics. \textbf{(a)} Each guiding wavefunction defines a fixed-node Hamiltonian whose ground state becomes the next guide; the energies are non-increasing in $k$. \textbf{(b--d)} Starting from the \textbf{(b)} \textit{ground state sign chamber:} Iterative fixed node flows to the true ground state from any point. 
\textbf{(c)} \textit{excited state sign chamber}: Iterative fixed node generically has repulsive directions (orange) and may have attractive directions (gray); different initial trial wavefunction can approach distinct boundary fixed points. \textbf{(d)} \textit{a chamber with no eigenstate}: Trajectories starting near a generic point (gray dot) approach a boundary fixed point. Stars mark interior fixed points and open circles mark boundary limits. A fixed point on the surviving configurations is an eigenvector of the restricted Hamiltonian $H_{SS}$.}
  \label{fig:geometry}
\end{figure*}

While the sign-structure of $\Psi_T$  is not sufficient for lattice fixed-node, it is true that the output is at least as good as $\Psi_T$.  This leads to a natural question:  what is the flow of a process where one iterates FNDMC taking the sampled output wave-function as the input trial wave-function for the next iteration.  Notice that in such a map the sign structure stays preserved over the iterations.  In this work, we will address this question defining and looking at the flow of an iterative fixed node.  We focus on the nature of this flow independent of whether and how a practical algorithm could implement this iteration.\\ 

Formally, we can define our fixed node iteration at step $k+1$ as (schematic in Fig.\ref{fig:geometry-update})
\begin{equation}
    \psi_{k+1}=T(\psi_k):=\phi_0\!\left(F(\psi_k)\right),
    \label{eq:intro-map}
\end{equation}
where $F(\psi_k)$ is the lattice fixed-node Hamiltonian built from the trial wave-function $\psi_k$ and $\phi_0(F)$ is its normalized nondegenerate ground state for $F$.  
Notice that by construction the sign structure of each $\psi_k$ is fixed to that of  $\psi_0$.  
We study the fixed points of this map, their stability, and how convergence depends on the $\psi_0$ including determining whether the iterative process converges to the best state for a given sign-structure. 
The iterative process of Eq.\eqref{eq:intro-map} is non-increasing with respect to energy while preserving the sign structure.  Naively, one might anticipate that this then flows toward a fixed point wave-function with that sign structure.  This need not be the case as it's possible that the state flows towards a ``boundary''  of the sign-structure where some amplitudes go to zero (and hence without sign).  In fact, we show the fixed points of this iteration are both eigenstates of $H$  (which are at most one per sign-structure) as well as eigenstates of  principal submatrices of $H$, the latter which results in states with zero amplitude outside that respective submatrix.  Having identified the fixed points, we then consider the stability of these fixed points. 
This is done by showing that iterative lattice fixed-node projection given by Eq.~\eqref{eq:intro-map}, can be formulated as a self-consistent-field method for a nonlinear eigenvalue problem with eigenvector dependence (NEPv) \cite{jarlebring2014inverse,cai2018eigenvector,bai2022sharp,henning2025gross}. As shown in Fig.\ref{fig:geometry-ground} the ground state is perturbatively stable, moreover the boundaries of this chamber are repulsive. In contrast, Fig.~\ref{fig:geometry-excited} shows excited-state fixed points are unstable: the $q$-th excited state has precisely $q$ unstable directions associated with lower-energy states, directly relating the dynamics to its Morse index. 
By considering the `boundaries' of a given sign-structure, we further map out the geometry of the flow.  We find that the boundaries of the ground-state sign chamber are all repulsive;  therefore every initial $\psi_0$  with the sign structures of the correct ground state flow towards the true ground state under iteration.  States that start with a sign structure in an excited state sign chamber may flow towards the excited state fixed point for some time but eventually flow towards the walls converging towards an eigenstate of the reduced Hamiltonian.  Interestingly, we show that attractive walls in one sign chamber are always repulsive in another and so if we were to relax the dynamics to allow sign flips one could tunnel naturally from one chamber to another without being caught in the reduced Hamiltonian. Fig.\ref{fig:geometry-boundary} illustrates that initializations inside a generic sign chamber flows toward attractive boundary fixed points through zeroing out of certain configurations. Thus, despite the explicit amplitude dependence of a single lattice fixed-node construction, repeated self-consistency reconstructs the ground-state amplitudes from the correct sign structure alone.

\paragraph{Background}
Let $H=\sum_{a=1}^m H_a$ be an $n$-qubit Hamiltonian, real symmetric matrix in the computational basis $\{\ket{x}:{x\in V} \}$ indexed by a finite set $V$, and a unique non-degenerate  ground state $\phi_0(H)=\sum_x \phi_x \ket{x}$ such that $\phi_x\ne0$ for every $x$. We work with real Hamiltonians although ground states of gapped non-degenerate complex Hamiltonians on $n$ qubits can always be mapped to the ground state of a real Hamiltonian on $(n+1)$ qubits\cite{bravyi2023rapidly}. 
We define the ground state energy $E_0$ and spectral gap $\Delta_H := E_1-E_0 > 0$. Given the violating-edge set
$S_+(v)=\bigl\{\{x,z\}:x<z,\ v_xH_{xz}v_z>0\bigr\}$, the lattice fixed-node Hamiltonian is
\begin{equation}
\begin{aligned}
F(v)_{xz}&=
\begin{cases}
0, & \{x,z\}\in S_+(v),\\
H_{xz}, & \{x,z\}\notin S_+(v),
\end{cases}
&&x\ne z,\\
F(v)_{xx}&=H_{xx}+\sum_{(x,y) \in S_+(v)}H_{xy}\frac{v_y}{v_x}.
\end{aligned}
\label{eq:fixed-node}
\end{equation}
The Hamiltonian $F$ is sign-free with the deleted off-diagonal matrix elements  restored through a diagonal compensation. We use two identities from Ref. \cite{ten1995proof}:
\begin{equation}
  F(v)v=Hv,
  \qquad
  F(v)\succeq H.
  \label{eq:structural-identities}
\end{equation}
We define a chamber as the set of states with a fixed sign structure set by the sign vector $s$ and the chamber walls by states with zero amplitude values and consistent signs otherwise. Within a chamber, the violating-edge set is fixed and $F$ is smooth with respect to varying amplitudes. If $F$ is irreducible, $T$ preserves the guide's signs and full support at each iterative step. We call an approach to a limit $\psi_\star$ with proper support $S=\{x:\psi_\star(x)\ne0\}\subsetneq V$ \emph{support collapse}: amplitudes on $R=V\setminus S$ vanish asymptotically. The surviving Hamiltonian $H_{SS}$ restricts both indices of $H$ to $S$, and $F_S$ denotes its fixed-node construction. A boundary fixed point is a state whose nonzero restriction $\psi_S=\psi_\star|_S$ is fixed under the reduced map $\psi_S\mapsto\phi_0(F_S(\psi_S))$. The regularization used near zero amplitudes is described in Supplemental Material, Sec.~\ref{sec:regular}.
\begin{figure*}[t]
  \centering
  \includegraphics[width=\textwidth]{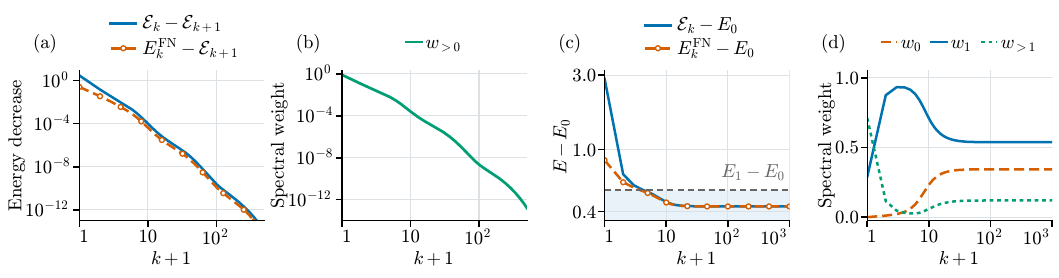}
  \caption{Fixed-node iteration in disordered $J_1$-$J_2$ chains: Initial trial wavefunction has \textbf{(a,b)} ground-state signs, $L=12$; \textbf{(c,d)} first-excited-state signs, $L=10$. Panel \textbf{(a)} shows the physical energy decrease $\mathcal{E}_k-\mathcal{E}_{k+1}$ and the gap $E_k^{\rm FN}-\mathcal{E}_{k+1}$, where $\mathcal{E}_k=\langle\psi_k,H\psi_k\rangle$ and $E_k^{\rm FN}=\lambda_0(F(\psi_k))$. Panel \textbf{(c)} shows $\mathcal{E}_k-E_0$ and $E_k^{\rm FN}-E_0$ crossing below the dashed reference $E_1-E_0$ into the shaded region. Panels \textbf{(b,d)} show spectral weights $w_n=|\langle\phi_n,\psi_k\rangle|^2$, with $w_{>r}=\sum_{n>r}w_n$. Decaying $w_{>0}$ shows $\psi_k \rightarrow \phi_0$, while growing $w_0$ in \textbf{(d)} reveals eigenstates acquire a finite spectral weight.}
  \label{fig:numerics}
\end{figure*}
\paragraph{Fixed points and walls} 
At a full-support fixed point $\psi_\star$, Eq.~\eqref{eq:structural-identities} gives $H\psi_\star=F(\psi_\star)\psi_\star=E_\star\psi_\star$, where $E_\star=\lambda_0(F(\psi_\star))$. From this it follows that all eigenstates of $H$ are fixed points for the map $T$.

For a boundary fixed point, the same identity on $S$ gives $H_{SS}\psi_S=\lambda\psi_S$, with $\lambda=\lambda_0(F_S(\psi_S))$. Its zero extension is an eigenstate of $H$ exactly when $H_{RS}\psi_S=0$ (Supplemental Material, Sec.~\ref{sec:boundary-supp}, Theorem~\ref{thm:boundary-supp}). Thus the full-support eigenstates of $H$ and the eigenstates of $H_{SS}$ with full support on $S$ exhaust the fixed points of the full and reduced maps, subject to this simplicity assumption. It should be noted that  ground-state of the reduced Hamiltonian $\psi_S$ can be written as a linear combination of eigenstates of the full Hamiltonian $H$. The perturbative stability of fixed points against perturbations of the non-zero amplitudes is analyzed below and expanded upon in the Supplemental Material, Secs.~\ref{sec:linearization-supp} and~\ref{sec:morse-supp}.

Since support collapse depends on the sign chamber from which the boundary is approached, a boundary fixed point may be stable under perturbations of its nonzero amplitudes but unstable when a small amplitude is introduced on a missing configuration. Let $b$ be a normalized reduced fixed point with one missing component, $b_r=0$, and approach it from either adjacent chamber with guiding wavefunction proportional to $b+\sigma z\ket{r}$, where $z>0$ and $\sigma=\pm1$. We define the transverse multiplier $\kappa_\sigma$ by $\sigma T(b+\sigma z\ket{r})_r=\kappa_\sigma z+O(z^2)$. If the reduced fixed-node ground state is simple and both deleted and retained incident couplings contribute, exchanging the two sides gives $\kappa_\sigma\kappa_{-\sigma}=1$ (Supplemental Material, Sec.~\ref{sec:wall-reciprocity-supp}). Hence a wall that attracts from one side, $0<\kappa_\sigma<1$, repels from the other. Such attraction can drive support collapse in an incorrect-sign chamber.

For normalized iterates, we define the physical energy $\mathcal E_k=\langle\psi_k,H\psi_k\rangle$ and the fixed-node energy $E_k^{\mathrm{FN}}=\lambda_0(F(\psi_k))$. The variational principle and Eq.~\eqref{eq:structural-identities} give
\begin{equation}
 \mathcal E_{k+1}\le E_k^{\mathrm{FN}}\le\mathcal E_k.
 \label{eq:energy-descent}
\end{equation}
In the ground-state chamber, $\phi_0$ is the only interior fixed point, since no other eigenstate of $H$ can share its signs and remain orthogonal to it. The contradiction proof in Supplemental Material, Sec.~\ref{sec:ground-boundary-supp}, shows that every reduced boundary fixed point admits an inward energy-lowering perturbation. Equation~\eqref{eq:energy-descent} therefore prevents such a point from attracting a full neighborhood of interior guiding wavefunction. At a regular wall with one missing component, the same argument gives $\kappa_\sigma>1$ on the ground-state side. If the fixed-node operators remain irreducible and the trajectory is uniformly interior, $\inf_{k,x}|\psi_{k,x}|>0$, energy descent implies convergence to $\phi_0$ (Supplemental Material, Sec.~\ref{sec:correct-sign-supp}, Theorem~\ref{thm:correct-sign-supp}).

\paragraph{Stable/Unstable fixed points}
To determine whether a fixed point attracts nearby states, we linearize one iteration of the map. This is the same local construction used for a self-consistent-field iteration of an eigenvector-dependent nonlinear eigenvalue problem~\cite{cai2018eigenvector,bai2022sharp}. Let $\psi_\star$ be a normalized full-support fixed point, write $F_\star=F(\psi_\star)$, and denote its fixed-node energy by $E_\star$. Assume that $\psi_\star$ is the nondegenerate ground state of $F_\star$. The identity $F(\psi_\star)\psi_\star=H\psi_\star$ then gives $H\psi_\star=E_\star\psi_\star$. Since $\psi_\star$ has full support, sufficiently small perturbations preserve its signs and violating-edge set, so $F$ varies smoothly nearby. A nearby normalized guiding wavefunction has the form $\psi=\psi_\star+\eta+O(\|\eta\|^2)$, where $\langle\psi_\star,\eta\rangle=0$. Thus the first-order perturbation $\eta$ lies in the tangent space $\mathcal T_{\psi_\star}=\{\eta:\langle\psi_\star,\eta\rangle=0\}$, with orthogonal projector $P_\perp=I-|\psi_\star\rangle\langle\psi_\star|$, where $I$ is the identity. On this space, we define the shifted physical and fixed-node operators $A=P_\perp(H-E_\star I)P_\perp$ and $B=P_\perp(F_\star-E_\star I)P_\perp$. The quadratic form $\langle\eta,A\eta\rangle$ gives the physical energy change to second order. The operator $B$ is positive definite because $E_\star$ is the nondegenerate lowest eigenvalue of $F_\star$. The fixed-node inequality $F_\star\succeq H$ also gives $A\preceq B$. We take $B^{-1}$ only on $\mathcal T_{\psi_\star}$. Let $\delta F$ denote the first-order change in $F$ induced by $\eta$. If we differentiate $F(\psi)\psi=H\psi$ we get: $\delta F\psi_\star=(H-F_\star)\eta$. If $\eta'$ is the first-order change in the updated ground state, projecting the perturbed ground-state eigenvalue equation onto $\mathcal T_{\psi_\star}$ gives $B\eta'=-P_\perp\delta F\psi_\star=(B-A)\eta$. Solving for $\eta'$ determines the Jacobian $J_\star=DT_{\psi_\star}$, the derivative of $T$ at the fixed point. The tangent perturbation $\eta_k$ at iteration $k$ therefore obeys
\begin{equation}
 \begin{aligned}
 \eta_{k+1}&=J_\star\eta_k+O(\|\eta_k\|^2),\\
 J_\star&=B^{-1}(B-A)=I-B^{-1}A,
 \end{aligned}
 \label{eq:linearized-map}
\end{equation}
where $I$ now acts on the tangent space. This self-consistent-field linearization~\cite{cai2018eigenvector,bai2022sharp} is derived in detail in Supplemental Material, Sec.~\ref{sec:linearization-supp}. To determine the eigenvalues of $J_\star$, consider a nonzero tangent vector $u$ satisfying the generalized eigenvalue equation $Au=\mu Bu$. Eq. ~\eqref{eq:linearized-map} gives $J_\star u=(1-\mu)u$, so the same mode has Jacobian eigenvalue $\lambda=1-\mu$. All generalized eigenvalues $\mu$ are real because $A$ and $B$ are symmetric and $B$ is positive definite. To first order, the amplitude of this mode is multiplied by $\lambda$ at each step. The fixed point is locally attractive if the spectral radius $\rho(J_\star)=\max_\lambda|\lambda|$, the largest absolute Jacobian eigenvalue, is less than one. Any mode with $|\lambda|>1$ is unstable. For the unique ground state $\phi_0$ of $H$, take $\psi_\star=\phi_0$ and $E_\star=E_0$. Then $A$ is positive definite, so $0\prec A\preceq B$. Hence every generalized eigenvalue satisfies $0<\mu\le1$, so every Jacobian eigenvalue lies in $[0,1)$. The ground-state fixed point is therefore locally attractive. Now take $\psi_\star=\phi_q$ to be a full-support excited-state fixed point with nondegenerate energy $E_q$, and let $q$ count the eigenvalues of $H$ below $E_q$, including multiplicity. Here $A$ and $B$ are evaluated at $\psi_\star=\phi_q$ and $E_\star=E_q$, with the same assumption that $\phi_q$ is a nondegenerate ground state of $F(\phi_q)$. The eigenvalues of $A$ are the energy differences between the other eigenstates of $H$ and $\phi_q$, so exactly $q$ are negative. Because $B$ is positive definite, the generalized eigenvalue problem has the same number of negative eigenvalues as $A$ (see Theorem~\ref{thm:morse-supp} in the Supplemental Material). These $q$ values $\mu<0$ give $q$ Jacobian eigenvalues $\lambda>1$. For each expanding mode, $\langle u,Au\rangle=\mu\langle u,Bu\rangle<0$, so the perturbation lowers the physical energy to second order. The remaining generalized eigenvalues lie in $(0,1]$, since $E_q$ is nondegenerate and $A\preceq B$. Thus the excited-state fixed point has exactly $q$ unstable directions, and all remaining directions contract. If the fixed-node ground state is degenerate, $B$ is singular on the tangent space and this linearization does not apply. The analysis also assumes an exact ground-state solution at every fixed-node update, how stochastic or approximate solutions affect convergence remains open.

\begin{figure*}[t]
  \centering
  \begin{subfigure}[t]{0.30\textwidth}
    \caption{}\label{fig:stability-radius}
    \includegraphics[width=\linewidth]{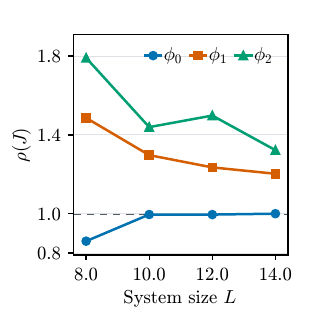}
  \end{subfigure}\hfill
  \begin{minipage}[t]{0.69\textwidth}
    \begin{subfigure}[t]{0.5\linewidth}
      \caption{}\label{fig:stability-ground}
    \end{subfigure}%
    \begin{subfigure}[t]{0.5\linewidth}
      \caption{}\label{fig:stability-excited}
    \end{subfigure}
    \includegraphics[width=\linewidth]{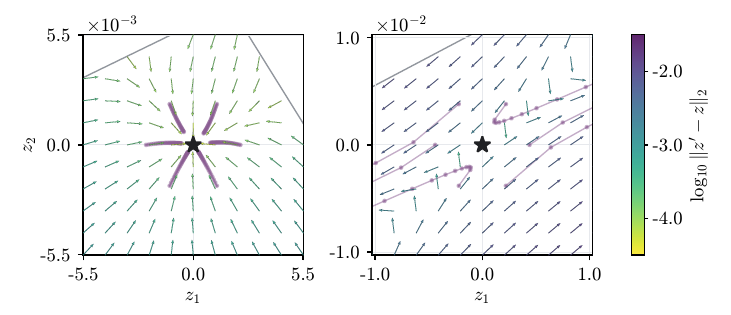}
  \end{minipage}
  \caption{Perturbative stability of the iterative fixed-node. \textbf{(a)} Jacobian spectral radius near the ground and first two excited states of disordered $J_1$-$J_2$ chains; the dashed line marks $\rho=1$. \textbf{(b,c)} Projected one-step updates near $\phi_0$ and $\phi_1$ of an eight-site triangular strip, with $\rho(J)=0.951$ and $1.404$, respectively. Arrows have uniform length and show the direction of $z'-z$; the color of the arrow corresponds to the magnitude, $\log_{10}\|z'-z\|_2$. Each panel's multiplier applies to both coordinates. Purple curves trace successive fixed-node iterates projected into the same axes. Gray lines mark component-zero boundaries in the displayed plane; stars mark the reference eigenstates. The pooled out-of-plane displacement fractions are $25.6\%$ and $33.9\%$, respectively.}
  \label{fig:phase-portraits-supp}
\end{figure*}

\paragraph{Numerics}
We will now numerically evaluate the iterative fixed node map and measure spectral radius of the Jacobian explicitly in the zero-magnetization sector of spin-$1/2$ models with Hamiltonian
\begin{equation}
 H=\sum_{(i,j)\in\mathcal E}J_{ij}\left(S_i^xS_j^x+S_i^yS_j^y
       +\Delta S_i^zS_j^z\right)+\sum_i h_iS_i^z,
 \label{eq:numerical-hamiltonian}
\end{equation}
where $\mathcal E$ is the set of coupled site pairs, $J_{ij}$ are exchange couplings, $\Delta$ is the exchange anisotropy, and the longitudinal fields $h_i$ are drawn independently from a Gaussian distribution with zero mean and standard deviation $h_{\rm dis}$. Fig.\ref{fig:numerics} compares two initial sign choices for open $J_1$-$J_2$ chains with $J_1=1$, $J_2=0.78$, $\Delta=1$, and $h_{\rm dis}=0.1$. Each run uses one disorder realization, random positive amplitudes assigned the signs of the target eigenstate, and a sign-preserving threshold guiding wavefunction with $\eps=10^{-8}$, defined in the Supplemental Material. For the $L=12$ ground-state signs [Fig.\ref{fig:numerics} (a,b)], the physical and fixed-node energies converge to $E_0$ and approach one another as the spectral weight concentrates on $\phi_0$. For the $L=10$ first-excited-state signs [Fig.\ref{fig:numerics} (c,d)], the initial guide is also projected orthogonal to $\phi_0$. A finite ground-state weight nevertheless develops, the energy falls approximately $0.12$ below $E_1$, and some amplitudes approach the regularization scale. The iteration thus lowers the energy through support collapse while retaining the initial signs. Fig.~\ref{fig:phase-portraits-supp}(a) calculates the spectral radius of the Jacobian $\rho(J)$ for chains with $L=8-14$, the same exchange couplings, and $h_{\rm dis}=0.05$. Its spectral radius is below one for the ground state and above one for the first two excited states at every displayed size.

Fig.\ref{fig:phase-portraits-supp}(b,c) show one-step updates near the ground and first excited states of an eight-site open triangular strip with $J_{ij}=1$, $\Delta=1$, and $h_{\rm dis}=0.12$, using the square-root guide with $\eps=10^{-8}$. We display the displacement under one fixed-node update by projecting onto a two-dimensional tangent plane. Let the columns of $U=(u_1,u_2)$ be orthonormal tangent directions at the reference state $\psi_\star$, and let $z=(z_1,z_2)$ denote coordinates in their plane. We use normalized guiding wavefunction $\psi(z)=(\psi_\star+Uz)$. For each sampled $z$ that preserves the reference signs, we compute $\widehat\psi=T(\psi(z))$ and define the full displacement in the same projective coordinates by $d(z)=\frac{\widehat\psi}{\langle\psi_\star,\widehat\psi\rangle}-\psi_\star-Uz.$ The updated plane coordinates are $z'=U^{\mathsf T}[\widehat\psi/\langle\psi_\star,\widehat\psi\rangle-\psi_\star]$, so the projected displacement is $z'-z=U^{\mathsf T}d(z)$. Arrows are normalized to a common length and show its direction; their color gives $\log_{10}\|z'-z\|_2$. The gray lines satisfy $\psi_\star(x)+[Uz]_x=0$ and purple curves follow a few trajectories explicitly, before projection onto the plane. The ground-state plane spans the two dominant Jacobian modes; the excited-state plane contains the lower-energy state $\phi_0$ and the dominant Jacobian mode. The current projection omits $d_\perp(z)=(I-UU^{\mathsf T})d(z)$, for better visualization. In the Supplemental Material, Sec.\ref{sec:exact-solve}, we construct an exactly solvable non-stoquastic Hamiltonian on a ring. The  iterative fixed-node dynamics reduces to a scalar map within a one-parameter family of positive initial trial wavefunction (Supplemental Material, Fig.~\ref{fig:solvable-ring-iteration}). Within this family, the amplitude bias decreases monotonically in magnitude to zero, recovering the uniform ground state.
\paragraph{Conclusion}
We have formulated an iterative lattice fixed-node procedure as a nonlinear eigenvalue problem with eigenvector dependency. The resulting fixed points divide naturally into two classes. Interior fixed points are exact eigenstates of the full Hamiltonian. Boundary fixed points are exact eigenstates of principal submatrices $H_{SS}$ defined by the surviving support $S$. In the ground-state sign chamber, we prove that walls are all repulsive, hence all the trajectories flow to ground state. Furthermore, in all sign chambers we show that points on walls at any boundary are attractive on one side, are repulsive on the other.  
This picture also suggests an algorithmic extension beyond the fixed-chamber map: once an amplitude falls below a chosen tolerance, one may remove that configuration temporarily or allow its inherited sign to be reassigned, thereby permitting motion between adjacent sign chambers. At a full-support fixed point with a simple fixed-node ground state, the tangent-space Jacobian determines local stability. In the ground-state sign sector, the fixed point is locally stable,  whereas excited states are provably unstable. Our work can be extended in many directions. SHDMC feeds stochastic fixed-node information back into a parametrized continuum trial state to improve its amplitudes and nodes~\cite{reboredo2009self,reboredo2012many,caffarel2026self}; here we isolate a deterministic lattice amplitude-repair mechanism within a fixed sign chamber. It would be interesting to develop similar algorithms for iterative lattice fixed-node. Conversely, Westerhout \textit{et al.} reconstruct signs from approximate amplitudes~\cite{westerhout2023many}, whereas our convergence theorem reconstructs the exact ground-state amplitudes from the exact signs under the stated hypotheses. This suggests a possible mechanism, where we bootstrap sign reconstruction with fixed-node amplitude updates. We explore this possibility for small systems in the Supplemental material, Sec. \ref{sec:glassysign}. The deterministic map studied here represents the ideal limit where each fixed-node ground state is computed exactly, whether the combined iteration converges with sampled configurations, finite walker populations, and approximate projection remains an open question. \\
\textit{Acknowledgments---} We gratefully acknowledge discussions with Daniel Belkin. 

\bibliographystyle{apsrev4-2}
\bibliography{fndmc}

@article{bravyi2023rapidly,
  doi = {10.22331/q-2023-11-07-1173},
  url = {https://doi.org/10.22331/q-2023-11-07-1173},
  title = {A rapidly mixing {M}arkov chain from any gapped quantum many-body system},
  author = {Bravyi, Sergey and Carleo, Giuseppe and Gosset, David and Liu, Yinchen},
  journal = {{Quantum}},
  issn = {2521-327X},
  publisher = {{Verein zur F{\"{o}}rderung des Open Access Publizierens in den Quantenwissenschaften}},
  volume = {7},
  pages = {1173},
  month = nov,
  year = {2023}
}

@article{jiang2023local,
  title = {Local Hamiltonian Problem with Succinct Ground State is {MA}-Complete},
  author = {Jiang, Jiaqing},
  journal = {PRX Quantum},
  volume = {6},
  number = {2},
  pages = {020312},
  year = {2025},
  doi = {10.1103/PRXQuantum.6.020312},
  url = {https://doi.org/10.1103/PRXQuantum.6.020312}
}

@article{ten1995proof,
  title = {Proof for an upper bound in fixed-node Monte Carlo for lattice fermions},
  author = {{ten Haaf}, D. F. B. and {van Bemmel}, H. J. M. and {van Leeuwen}, J. M. J. and {van Saarloos}, W. and Ceperley, D. M.},
  journal = {Phys. Rev. B},
  volume = {51},
  number = {19},
  pages = {13039--13045},
  year = {1995},
  doi = {10.1103/PhysRevB.51.13039},
  url = {https://doi.org/10.1103/PhysRevB.51.13039}
}

@article{boninsegni1995triangular,
  title = {Ground state of a triangular quantum antiferromagnet: Fixed-node Green-function Monte Carlo study},
  author = {Boninsegni, Massimo},
  journal = {Phys. Rev. B},
  volume = {52},
  number = {21},
  pages = {15304--15311},
  year = {1995},
  doi = {10.1103/PhysRevB.52.15304},
  url = {https://doi.org/10.1103/PhysRevB.52.15304}
}

@article{sorella1998green,
  title = {Green Function Monte Carlo with Stochastic Reconfiguration},
  author = {Sorella, Sandro},
  journal = {Phys. Rev. Lett.},
  volume = {80},
  number = {20},
  pages = {4558--4561},
  year = {1998},
  doi = {10.1103/PhysRevLett.80.4558},
  url = {https://doi.org/10.1103/PhysRevLett.80.4558}
}

@article{sorella2000green,
  title = {Green function Monte Carlo with stochastic reconfiguration: An effective remedy for the sign problem},
  author = {Sorella, Sandro and Capriotti, Luca},
  journal = {Phys. Rev. B},
  volume = {61},
  number = {4},
  pages = {2599--2612},
  year = {2000},
  doi = {10.1103/PhysRevB.61.2599},
  url = {https://doi.org/10.1103/PhysRevB.61.2599}
}

@misc{waite2025complexitysuccinctstatelocal,
      title={On the Complexity of the Succinct State Local Hamiltonian Problem}, 
      author={Gabriel Waite and Karl Lin},
      year={2025},
      eprint={2509.25821},
      archivePrefix={arXiv},
      primaryClass={quant-ph},
      url={https://arxiv.org/abs/2509.25821}, 
}

@book{kato1995perturbation,
  title={Perturbation Theory for Linear Operators},
  author={Kato, Tosio},
  year={1995},
  publisher={Springer}
}

@article{westerhout2023many,
  title={Many-body quantum sign structures as non-glassy Ising models},
  author={Westerhout, Tom and Katsnelson, Mikhail I and Bagrov, Andrey A},
  journal={Communications Physics},
  volume={6},
  number={1},
  pages={275},
  year={2023},
  publisher={Nature Publishing Group UK London},
  doi={10.1038/s42005-023-01388-6}
}

@article{henning2025gross,
  title={The Gross--Pitaevskii equation and eigenvector nonlinearities: numerical methods and algorithms},
  author={Henning, Patrick and Jarlebring, Elias},
  journal={SIAM Review},
  volume={67},
  number={2},
  pages={256--317},
  year={2025},
  publisher={SIAM},
  doi={10.1137/22M1516324}
}

@article{jarlebring2014inverse,
  title={An inverse iteration method for eigenvalue problems with eigenvector nonlinearities},
  author={Jarlebring, Elias and Kvaal, Simen and Michiels, Wim},
  journal={SIAM Journal on Scientific Computing},
  volume={36},
  number={4},
  pages={A1978--A2001},
  year={2014},
  doi={10.1137/130910014}
}

@article{cai2018eigenvector,
  title={On an eigenvector-dependent nonlinear eigenvalue problem},
  author={Cai, Yunfeng and Zhang, Lei-Hong and Bai, Zhaojun and Li, Ren-Cang},
  journal={SIAM Journal on Matrix Analysis and Applications},
  volume={39},
  number={3},
  pages={1360--1382},
  year={2018},
  doi={10.1137/17M115935X}
}

@article{bai2022sharp,
  title={Sharp estimation of convergence rate for self-consistent field iteration to solve eigenvector-dependent nonlinear eigenvalue problems},
  author={Bai, Zhaojun and Li, Ren-Cang and Lu, Ding},
  journal={SIAM Journal on Matrix Analysis and Applications},
  volume={43},
  number={1},
  pages={301--327},
  year={2022},
  publisher={SIAM},
  doi={10.1137/20M136606X}
}

@article{johnson1995principal,
  title = {Principal Submatrices, Geometric Multiplicities, and Structured Eigenvectors},
  author = {Johnson, Charles R. and Kroschel, Brenda K.},
  journal = {SIAM J. Matrix Anal. Appl.},
  volume = {16},
  number = {3},
  pages = {1004--1012},
  year = {1995},
  doi = {10.1137/S0895479894266568},
  url = {https://doi.org/10.1137/S0895479894266568}
}

@article{haemers1995interlacing,
  title = {Interlacing eigenvalues and graphs},
  author = {Haemers, Willem H.},
  journal = {Linear Algebra Appl.},
  volume = {226--228},
  pages = {593--616},
  year = {1995},
  doi = {10.1016/0024-3795(95)00199-2},
  url = {https://doi.org/10.1016/0024-3795(95)00199-2}
}

@article{reboredo2009self,
  title={Self-healing diffusion quantum Monte Carlo algorithms: Direct reduction of the fermion sign error in electronic structure calculations},
  author={Reboredo, Fernando A and Hood, Randolph Q and Kent, Paul RC},
  journal={Phys. Rev. B},
  volume={79},
  number={19},
  pages={195117},
  year={2009},
  publisher={APS},
  doi={10.1103/PhysRevB.79.195117}
}

@article{reboredo2012many,
  title={Many-body calculations of low-energy eigenstates in magnetic and periodic systems with self-healing diffusion Monte Carlo: Steps beyond the fixed phase},
  author={Reboredo, Fernando A.},
  journal={J. Chem. Phys.},
  volume={136},
  number={20},
  pages={204101},
  year={2012},
  doi={10.1063/1.4711023}
}

@article{PhysRevLett.128.040403,
  title = {Ruling Out Real-Valued Standard Formalism of Quantum Theory},
  author = {Chen, Ming-Cheng and Wang, Can and Liu, Feng-Ming and Wang, Jian-Wen and Ying, Chong and Shang, Zhong-Xia and Wu, Yulin and Gong, M. and Deng, H. and Liang, F.-T. and Zhang, Qiang and Peng, Cheng-Zhi and Zhu, Xiaobo and Cabello, Ad\'an and Lu, Chao-Yang and Pan, Jian-Wei},
  journal = {Phys. Rev. Lett.},
  volume = {128},
  issue = {4},
  pages = {040403},
  numpages = {5},
  year = {2022},
  month = {Jan},
  publisher = {American Physical Society},
  doi = {10.1103/PhysRevLett.128.040403},
  url = {https://link.aps.org/doi/10.1103/PhysRevLett.128.040403}
}

@article{cosentini1998phase,
  title = {Phase separation in the two-dimensional {Hubbard} model: A fixed-node quantum {Monte Carlo} study},
  author = {Cosentini, A. C. and Capone, M. and Guidoni, L. and Bachelet, G. B.},
  journal = {Phys. Rev. B},
  volume = {58},
  number = {22},
  pages = {R14685},
  year = {1998},
  doi = {10.1103/PhysRevB.58.R14685},
  url = {https://doi.org/10.1103/PhysRevB.58.R14685}
}

@article{tocchio2008backflow,
  title = {Role of backflow correlations for the nonmagnetic phase of the {$t$--$t'$ Hubbard} model},
  author = {Tocchio, Luca F. and Becca, Federico and Parola, Alberto and Sorella, Sandro},
  journal = {Phys. Rev. B},
  volume = {78},
  number = {4},
  pages = {041101},
  year = {2008},
  doi = {10.1103/PhysRevB.78.041101},
  url = {https://doi.org/10.1103/PhysRevB.78.041101}
}

@article{leblanc2015solutions,
  title = {Solutions of the Two-Dimensional {Hubbard} Model: Benchmarks and Results from a Wide Range of Numerical Algorithms},
  author = {LeBlanc, J. P. F. and Antipov, Andrey E. and Becca, Federico and Bulik, Ireneusz W. and Chan, Garnet Kin-Lic and Chung, Chia-Min and Deng, Youjin and Ferrero, Michel and Henderson, Thomas M. and Jim{\'e}nez-Hoyos, Carlos A. and Kozik, E. and Liu, Xuan-Wen and Millis, Andrew J. and Prokof'ev, N. V. and Qin, Mingpu and Scuseria, Gustavo E. and Shi, Hao and Svistunov, B. V. and Tocchio, Luca F. and Tupitsyn, I. S. and White, Steven R. and Zhang, Shiwei and Zheng, Bo-Xiao and Zhu, Zhenyue and Gull, Emanuel},
  journal = {Phys. Rev. X},
  volume = {5},
  number = {4},
  pages = {041041},
  year = {2015},
  doi = {10.1103/PhysRevX.5.041041},
  url = {https://doi.org/10.1103/PhysRevX.5.041041}
}

@misc{loehr2025enhancing,
  title = {Enhancing Neural Network Backflow},
  author = {Loehr, Kieran and Clark, Bryan K.},
  year = {2025},
  eprint = {2510.26906},
  archivePrefix = {arXiv},
  primaryClass = {cond-mat.str-el},
  doi = {10.48550/arXiv.2510.26906},
  url = {https://arxiv.org/abs/2510.26906}
}

@article{GMFCjj,
  title = {Green's function Monte Carlo combined with projected entangled pair state approach to the frustrated ${J}_{1}\text{\ensuremath{-}}{J}_{2}$ Heisenberg model},
  author = {Lin, He-Yu and Guo, Yibin and He, Rong-Qiang and Xie, Z. Y. and Lu, Zhong-Yi},
  journal = {Phys. Rev. B},
  volume = {109},
  issue = {23},
  pages = {235133},
  numpages = {11},
  year = {2024},
  month = {Jun},
  publisher = {American Physical Society},
  doi = {10.1103/PhysRevB.109.235133},
  url = {https://link.aps.org/doi/10.1103/PhysRevB.109.235133}
}

@article{introtofndmc,
    author = {Annarelli, Alfonso and Alfè, Dario and Zen, Andrea},
    title = {A brief introduction to the diffusion Monte Carlo method and the fixed-node approximation},
    journal = {The Journal of Chemical Physics},
    volume = {161},
    number = {24},
    pages = {241501},
    year = {2024},
    month = {12},
    issn = {0021-9606},
    doi = {10.1063/5.0232424},
    url = {https://doi.org/10.1063/5.0232424},
}

@article{fndmcmath1,
    author = {Caffarel, Michel and Del Moral, Pierre and de Montella, Luc},
    title = {On the mathematical foundations of diffusion Monte Carlo},
    journal = {Journal of Mathematical Physics},
    volume = {66},
    number = {1},
    pages = {013301},
    year = {2025},
    month = {01},
    issn = {0022-2488},
    doi = {10.1063/5.0202800},
    url = {https://doi.org/10.1063/5.0202800},
}

@article{fndmcmath2,
author = {Canc{\`e}s, Eric and Jourdain, Benjamin and Leli{\`e}vre, Tony},
  title = {Quantum Monte Carlo simulations of fermions: a mathematical analysis of the fixed-node approximation},
journal = {Mathematical Models and Methods in Applied Sciences},
volume = {16},
number = {09},
pages = {1403-1440},
year = {2006},
doi = {10.1142/S0218202506001583},

URL = { https://doi.org/10.1142/S0218202506001583},
}

@article{lrdmc,
  title = {Diffusion Monte Carlo Method with Lattice Regularization},
  author = {Casula, Michele and Filippi, Claudia and Sorella, Sandro},
  journal = {Phys. Rev. Lett.},
  volume = {95},
  issue = {10},
  pages = {100201},
  numpages = {4},
  year = {2005},
  month = {Sep},
  publisher = {American Physical Society},
  doi = {10.1103/PhysRevLett.95.100201},
  url = {https://link.aps.org/doi/10.1103/PhysRevLett.95.100201}
}

@article{lrdmc2,
  title = {Beyond the locality approximation in the standard diffusion Monte Carlo method},
  author = {Casula, Michele},
  journal = {Phys. Rev. B},
  volume = {74},
  issue = {16},
  pages = {161102(R)},
  numpages = {4},
  year = {2006},
  month = {Oct},
  publisher = {American Physical Society},
  doi = {10.1103/PhysRevB.74.161102},
  url = {https://link.aps.org/doi/10.1103/PhysRevB.74.161102}
}

@article{signproblemreview,
   author = "Li, Zi-Xiang and Yao, Hong",
   title = "Sign-Problem-Free Fermionic Quantum Monte Carlo: Developments and Applications", 
   journal= "Annual Review of Condensed Matter Physics",
   year = "2019",
   volume = "10",
   number = "Volume 10, 2019",
   pages = "337-356",
   doi = "https://doi.org/10.1146/annurev-conmatphys-033117-054307",
   url = "https://www.annualreviews.org/content/journals/10.1146/annurev-conmatphys-033117-054307",
   publisher = "Annual Reviews",
   issn = "1947-5462",
   type = "Journal Article",
  }

@article{signprob_majorana,
  title = {Majorana Positivity and the Fermion Sign Problem of Quantum Monte Carlo Simulations},
  author = {Wei, Z. C. and Wu, Congjun and Li, Yi and Zhang, Shiwei and Xiang, T.},
  journal = {Phys. Rev. Lett.},
  volume = {116},
  issue = {25},
  pages = {250601},
  numpages = {5},
  year = {2016},
  month = {Jun},
  publisher = {American Physical Society},
  doi = {10.1103/PhysRevLett.116.250601},
  url = {https://link.aps.org/doi/10.1103/PhysRevLett.116.250601}
}

@article{signprob_science,
author = {Dominik Hangleiter  and Ingo Roth  and Daniel Nagaj  and Jens Eisert },
title = {Easing the Monte Carlo sign problem},
journal = {Science Advances},
volume = {6},
number = {33},
pages = {eabb8341},
year = {2020},
doi = {10.1126/sciadv.abb8341},
URL = {https://www.science.org/doi/abs/10.1126/sciadv.abb8341},
}

@article{caffarel2026self,
  title={Self-Healing Diffusion Monte Carlo applied to a simple fermionic model: A critical assessment of the method},
  author={Caffarel, Michel and Pinar, Manon and Scemama, Anthony},
  journal={arXiv preprint arXiv:2609.01301},
  year={2026}
}

\clearpage

\onecolumngrid
\appendix
\setcounter{secnumdepth}{1}
\setcounter{figure}{0}
\renewcommand{\thefigure}{S\arabic{figure}}

\begin{center}
{\bfseries Supplemental Material for}\\[0.5em]
{\bfseries\textit{Geometric View of Iterative Fixed-Node Dynamics}}\\[1.0em]

Pranav Kairon$^1$ and Bryan K. Clark$^{1}$\\[0.5em]

{\small
$^1$The Anthony J. Leggett Institute for Condensed Matter Theory and IQUIST and
NCSA Center for Artificial Intelligence Innovation and Department of Physics,
University of Illinois at Urbana-Champaign, Illinois 61801, USA
}\\[0.5em]
\end{center}

\begin{quote}
This Supplemental Material provides the details underlying the results in the Letter. We first re-derive the identities that govern the fixed-node map and then establish its convergence and local stability properties. We also explain the regularization used when components of the state approach zero and present additional numerical tests. We work on a finite configuration space $V$, with the Hamiltonian $H$ represented by a real symmetric matrix. Unless stated otherwise, all vectors are real and normalized.
\end{quote}

\tableofcontents

\section{Fixed-node identities and variational inequalities}

For a full-support vector $v$, we define the unordered violating-edge set
\begin{equation}
 S_+(v)=\bigl\{\{x,z\}:x<z,\ v_xH_{xz}v_z>0\bigr\}.
 \label{eq:supp-violating}
\end{equation}
The fixed-node operator is
\begin{equation}
\bigl(F(v)\bigr)_{xz}=
\begin{cases}
0, & x\ne z,\ \{x,z\}\in S_+(v),\\
H_{xz}, & x\ne z,\ \{x,z\}\notin S_+(v),\\
H_{xx}+\displaystyle\sum_{\substack{y\ne x\\ \{x,y\}\in S_+(v)}}
H_{xy}\dfrac{v_y}{v_x}, & x=z.
\end{cases}
\label{eq:supp-fixed-node}
\end{equation}
It is invariant under nonzero scalar rescaling of the guiding wavefunction.

\begin{lemma}
\label{lem:Fv-Hv}
For every full-support $v$,
\begin{equation}
 F(v)v=Hv.
 \label{eq:Fv-Hv}
\end{equation}
\end{lemma}

\begin{proof}
For a fixed row $x$, every deleted term $H_{xy}v_y$ is restored by the diagonal contribution $H_{xy}(v_y/v_x)v_x$. All nonviolating off-diagonal terms are unchanged, so the two matrix-vector products agree componentwise.
\end{proof}

\begin{lemma}
\label{lem:loewner}
For every full-support $v$,
\begin{equation}
 F(v)-H
 =\sum_{\{x,z\}\in S_+(v)}
 \frac{H_{xz}}{v_xv_z}
 (v_z|x\rangle-v_x|z\rangle)
 (v_z\langle x|-v_x\langle z|)
 \succeq0.
 \label{eq:penalty}
\end{equation}
\end{lemma}

\begin{proof}
For each violating edge, the rank-one term in Eq.~\eqref{eq:penalty} contributes $H_{xz}v_z/v_x$ and $H_{xz}v_x/v_z$ to the two diagonal entries and $-H_{xz}$ to the two off-diagonal entries. These are exactly the changes made by Eq.~\eqref{eq:supp-fixed-node}. Moreover,
$H_{xz}/(v_xv_z)>0$ because $v_xH_{xz}v_z>0$. Every term is therefore positive semidefinite.
\end{proof}

Equations~\eqref{eq:Fv-Hv} and \eqref{eq:penalty} imply, for any normalized $u$,
\begin{equation}
 \langle u,Hu\rangle\le \langle u,F(v)u\rangle,
 \qquad
 \langle v,F(v)v\rangle=\langle v,Hv\rangle.
 \label{eq:variational-identities}
\end{equation}
We define physical energy as $\mathcal E_k=\langle\psi_k,H\psi_k\rangle$, and the fixed-node energy as $E_k^{\mathrm{FN}}=\lambda_0(F(\psi_k))$. The variational principle and Lemma~\ref{lem:Fv-Hv} give
\begin{align}
 E_k^{\mathrm{FN}}
 &=\langle\psi_{k+1},F(\psi_k)\psi_{k+1}\rangle\nonumber\\
 &\le\langle\psi_k,F(\psi_k)\psi_k\rangle
 =\mathcal E_k.
\end{align}
Lemma~\ref{lem:loewner} gives the other inequality,
\begin{equation}
 \mathcal E_{k+1}
 \le\langle\psi_{k+1},F(\psi_k)\psi_{k+1}\rangle
 =E_k^{\mathrm{FN}}.
\end{equation}
Hence $\mathcal E_{k+1}\le\mathcal E_k$ and
$E_{k+1}^{\mathrm{FN}}\le\mathcal E_{k+1}\le E_k^{\mathrm{FN}}$.
Both sequences are bounded below by $\lambda_0(H)$ and hence the two limits converge to the same value.

\begin{lemma}[Approximate ground-state convergence]
\label{lem:approximate-ground}
Let $A_j\to A$ in operator norm, with all matrices real symmetric and $\lambda_0(A)$ simple. If normalized $x_j$ satisfy
\begin{equation}
 x_j^\mathsf{T}A_jx_j-\lambda_0(A_j)\longrightarrow0,
\end{equation}
then, after consistent global sign choices, $x_j\to\phi_0(A)$.
\end{lemma}

\begin{proof}
Let $u_j=\phi_0(A_j)$. For sufficiently large $j$, the spectral gap $\Delta(A_j)$ is at least half the positive limiting gap. Write
$x_j=\alpha_ju_j+\sqrt{1-\alpha_j^2}\,w_j$, with $w_j\perp u_j$, choosing the sign so that $\alpha_j\ge0$. Then
\begin{equation}
 x_j^\mathsf{T}A_jx_j-\lambda_0(A_j)
 \ge(1-\alpha_j^2)\Delta(A_j),
\end{equation}
so $\alpha_j\to1$. By continuity of a simple eigenvector under symmetric perturbations we get $u_j\to\phi_0(A)$~\cite{kato1995perturbation}.
\end{proof}

\section{Convergence in the ground-state sign chamber}
\label{sec:correct-sign-supp}

\begin{theorem}[Correct-sign convergence]
\label{thm:correct-sign-supp}
Let $H$ have a unique normalized ground state $\phi_0$ with full support. Assume that
\begin{enumerate}
 \item $\sgn(\psi_0)=\sgn(\phi_0)$;
 \item every $F(\psi_k)$ is irreducible and has a unique ground state; and
 \item the exact iteration remains uniformly interior,
 $\inf_{k,x}|\psi_{k,x}|>0$, equivalently every accumulation point has full support.
\end{enumerate}
Then
\begin{equation}
 \psi_k\longrightarrow\phi_0,
 \qquad
 \mathcal E_k\downarrow E_0,
 \qquad
 E_k^{\mathrm{FN}}\downarrow E_0,
 \label{eq:correct-sign-limit}
\end{equation}
where $E_0=\lambda_0(H)$.
\end{theorem}

\begin{proof}
Given that $\sgn(\psi_k)=\sgn(\phi_0)$ for every $k$ and the fact that $\psi_l$ lies on the compact unit sphere of a finite-dimensional state space, \(\{\psi_k\}\) has at least one convergent subsequence. Let \(\{\psi_{k_j}\}_{j\ge1}\) be any such subsequence, with \(\psi_{k_j}\to\psi_\infty\). The fixed off-diagonal graph in the chamber remains irreducible, so the limiting fixed-node ground state is simple. Set $A_j=F(\psi_{k_j})$ and $x_j=\psi_{k_j}$. Lemma~\ref{lem:Fv-Hv} and the vanishing energy slack established above give
\begin{equation}
 x_j^\mathsf{T}A_jx_j-\lambda_0(A_j)
 =\mathcal E_{k_j}-E_{k_j}^{\mathrm{FN}}\longrightarrow0.
\end{equation}
Lemma~\ref{lem:approximate-ground} therefore implies
\begin{equation}
 \psi_\infty=\phi_0(F(\psi_\infty)),
\end{equation}
so $\psi_\infty$ is a full-support fixed point. Lemma~\ref{lem:Fv-Hv} then yields
$H\psi_\infty=E_\infty\psi_\infty$. A distinct eigenvector of the symmetric matrix $H$ is orthogonal to $\phi_0$, but two full-support vectors with the same componentwise signs have a strictly positive inner product. Hence $\psi_\infty=\phi_0$. Because the argument applies to every convergent subsequence, $\phi_0$ is the unique accumulation point and the full sequence converges. The energy limits follow.
\end{proof}

\begin{corollary}[Overlap bound]
\label{cor:overlap-supp}
Let $\Delta_H=\lambda_1(H)-\lambda_0(H)>0$. For every normalized $\psi_k$,
\begin{equation}
 1-|\langle\phi_0,\psi_k\rangle|^2
 \le\frac{\mathcal E_k-E_0}{\Delta_H}.
 \label{eq:overlap-supp}
\end{equation}
Under the assumptions of Theorem~\ref{thm:correct-sign-supp}, the left-hand side vanishes.
\end{corollary}

\begin{proof}
We write $\psi_k=\alpha_k\phi_0+\sqrt{1-|\alpha_k|^2}\,w_k$, with $w_k\perp\phi_0$. The spectral gap gives
$\mathcal E_k\ge E_0+(1-|\alpha_k|^2)\Delta_H$.
\end{proof}

\section{Boundary fixed points and support collapse}
\label{sec:boundary-supp}
Let $\psi_\star$ have support $S=\supp(\psi_\star)\subsetneq V$ and complement $R=V\setminus S$. Hence, 
\begin{equation}
 H=\begin{pmatrix}H_{SS}&H_{SR}\\H_{RS}&H_{RR}\end{pmatrix}.
\end{equation}

In the weighted configuration graph of $H$, the set $S$ is the induced subgraph on the configurations whose amplitudes survive, and $H_{SS}$ is the corresponding principal-submatrix operator. The relation between zero coordinates and principal-submatrix eigenvectors is studied in Ref.~\cite{johnson1995principal}; eigenvalue interlacing describes how the spectrum changes under such vertex deletion \cite{haemers1995interlacing}.

\begin{theorem}[Reduced-support fixed points]
\label{thm:boundary-supp}
Suppose an iteration has a subsequence converging to $\psi_\star$. If the nonzero restriction $\psi_S$ is a fixed point of the reduced fixed-node iteration constructed from $H_{SS}$, then
\begin{equation}
 H_{SS}\psi_S=\lambda\psi_S,
 \qquad
 \lambda=\lambda_0(F_S(\psi_S)).
 \label{eq:HSS-eigenvector}
\end{equation}

\end{theorem}

\begin{proof}
Reduced self-consistency gives
$F_S(\psi_S)\psi_S=\lambda_0(F_S(\psi_S))\psi_S$. Applying Lemma~\ref{lem:Fv-Hv} to the restricted Hamiltonian gives Eq.~\eqref{eq:HSS-eigenvector}. Finally,
\begin{equation}
 H\begin{pmatrix}\psi_S\\0\end{pmatrix}
 =\begin{pmatrix}H_{SS}\psi_S\\H_{RS}\psi_S\end{pmatrix},
\end{equation}
which proves the equivalence.
\end{proof}

\subsection{Ground-state boundary repulsion}
\label{sec:ground-boundary-supp}
The following contradiction argument shows why a reduced-support fixed point in the ground-state chamber admits an energy-lowering perturbation toward larger support.
Let $H$ be real symmetric with a unique normalized full-support ground state $\phi_0$ and energy $E_0$. Set $s_x=\sgn(\phi_{0,x})$ and define the closed normalized sign chamber
\begin{equation}
 \mathcal C_s=\{u\in\mathbb R^{|V|}:\|u\|_2=1,\ s_xu_x\ge0\ \text{for all }x\in V\}.
\end{equation}
Let $b\in\mathcal C_s$ have proper support $S=\supp(b)\subsetneq V$, complement $R=V\setminus S$, and energy $E_b=\mathcal E(b)$. If $H_{SS}b_S=E_b b_S$, then some $x\in R$ satisfies
\begin{equation}
 \left.\frac{d}{dt}\mathcal E\!\left(
 \frac{b+t s_x\ket{x}}{\sqrt{1+t^2}}\right)\right|_{t=0^+}
 =2s_x(Hb)_x<0.
 \label{eq:ground-boundary-descent}
\end{equation}
Consequently, $\mathcal E$ has no constrained local minimum on the boundary of $\mathcal C_s$. This can be shown since, $b$ has a zero component and the unique ground state has full support, hence the variational principle gives $E_b>E_0$. Their common signs imply $\langle\phi_0,b\rangle>0$. For $x\in R$, the perturbed state in Eq.~\eqref{eq:ground-boundary-descent} is normalized and remains in $\mathcal C_s$ for $t\ge0$. Upon differentiating its Rayleigh energy at $t=0$ we get $2s_x(Hb)_x$. Suppose, for contradiction, that every such derivative is nonnegative. The support eigenvalue equation makes $(H-E_b)b$ vanish on $S$, while $b$ vanishes on $R$. Therefore,
\begin{align}
 0&>(E_0-E_b)\langle\phi_0,b\rangle\nonumber\\
  &=\langle\phi_0,(H-E_b)b\rangle\nonumber\\
  &=\sum_{x\in R}|\phi_{0,x}|\,s_x(Hb)_x\ge0,
 \label{eq:ground-boundary-contradiction}
\end{align}
which is impossible. Thus at least one missing component has the negative derivative in Eq.~\eqref{eq:ground-boundary-descent}. Finally, any constrained boundary local minimum is stationary under normalized variations within its support, so it must satisfy $H_{SS}b_S=\mathcal E(b)b_S$. The negative derivative excludes such a minimum.
\subsection{Directional stability of a single-component wall}
\label{sec:wall-reciprocity-supp}
Let $b$ be normalized with $b_r=0$ and nonzero components on $S=V\setminus\{r\}$, and assume $b_S$ is the simple ground state of $F_S(b_S)$. For the guiding wavefunction, $v=b+\sigma z\ket{r}$, where $z>0$ and $\sigma=\pm1$, define the positive incident weights
\begin{equation}
 \alpha_\sigma=\sum_{\substack{j\in S\\\sigma H_{rj}b_j>0}}\sigma H_{rj}b_j,
 \qquad
 \beta_\sigma=-\sum_{\substack{j\in S\\\sigma H_{rj}b_j<0}}\sigma H_{rj}b_j.
 \label{eq:wall-incident-weights}
\end{equation}
We call the wall regular when $\alpha_\sigma,\beta_\sigma>0$. The deleted edges give $F(v)_{rr}=H_{rr}+\alpha_\sigma/z$, while the surviving block tends to $F_S(b_S)$. If we eliminate the diverging $r$-diagonal entry we get an $O(z)$ perturbation of this reduced operator. Since the ground state is simple, we have $\widehat b_S=b_S+O(z)$ for $\widehat b=T(v)$, with a bounded eigenvalue $\lambda_z$. The $r$th row of the eigenvalue equation reads
\begin{equation}
 \left(H_{rr}+\frac{\alpha_\sigma}{z}-\lambda_z\right)\sigma\widehat b_r
 =-\sigma\sum_{\substack{j\in S\\\sigma H_{rj}b_j<0}}H_{rj}\widehat b_j
 =\beta_\sigma+O(z).
 \label{eq:wall-row-equation}
\end{equation}
Consequently,
\begin{equation}
 \sigma\widehat b_r=\kappa_\sigma z+O(z^2),
 \qquad \kappa_\sigma=\frac{\beta_\sigma}{\alpha_\sigma},
 \qquad \kappa_\sigma\kappa_{-\sigma}=1,
 \label{eq:wall-reciprocity}
\end{equation}
because changing $\sigma$ interchanges $\alpha_\sigma$ and $\beta_\sigma$. In the ground-state chamber, we have  $\sigma(Hb)_r=\alpha_\sigma-\beta_\sigma<0$; hence $\kappa_\sigma>1$. This transverse expansion applies to a regular wall with one missing component. Thus for several missing components, we establish an inward energy-lowering direction without asserting the reciprocal formula.
\section{Regularized numerical map}
\label{sec:regular}
The exact operator contains ratios $v_y/v_x$ and is undefined when a guiding wavefunction amplitude vanishes. For numerical work, we fix inherited signs $\sigma_x\in\{\pm1\}$ at zero and define
\begin{equation}
 \widehat{\sgn}_{\sigma_x}(a)=
 \begin{cases}
 \sgn(a),&a\ne0,\\
 \sigma_x,&a=0,
 \end{cases}
 \qquad
 g_\eps(\psi)_x=\widehat{\sgn}_{\sigma_x}(\psi_x)
 \sqrt{\psi_x^2+\eps^2}.
 \label{eq:floor-supp}
\end{equation}

Since $F(cv)=F(v)$, the floored guide need not be normalized before constructing $F$. Any plotted guiding wavefunction energy is understood as the Rayleigh quotient
$E_H[g]=g^\mathsf{T}Hg/(g^\mathsf{T}g)$. In Fig.~\ref{fig:triangular-supp} we use Eq.~\eqref{eq:floor-supp}. The exact variational comparison is $E_H[T(g)]\le E_0^{\rm FN}[g]\le E_H[g]$; replacing the last guide energy by $E_H[\psi]$ need not preserve the inequality when $g\ne\psi$. In the ground-state chain run all amplitudes stay above the threshold, so the two energy curves approach one another under the exact iteration, rather than coinciding by construction. For a fixed full-support $\psi$, uniformly away from the chamber boundary,
\begin{equation}
 g_\eps(\psi)_x
 =\psi_x+\frac{\eps^2}{2\psi_x}+O(\eps^4).
 \label{eq:floor-expansion}
\end{equation}
Thus $F(g_\eps(\psi))=F(\psi)+O(\eps^2)$. If the relevant fixed-node ground state is simple with a uniform gap and no component approaches zero, standard eigenvector perturbation theory further gives
\begin{equation}
 T_\eps(\psi):=\phi_0(F(g_\eps(\psi)))
 =T(\psi)+O(\eps^2).
 \label{eq:regularized-map-error}
\end{equation}
The floor defines a family of numerical maps rather than changing the exact theorems above. A useful support diagnostic, when both limits exist, is
\begin{equation}
 n_0(C)=\lim_{\eps\to0}\lim_{k\to\infty}
 \left|\left\{x:|\psi_{k,\eps,x}|\le C\eps\right\}\right|,
 \label{eq:nzero-supp}
\end{equation}
for a fixed threshold factor $C>0$. A stable nonzero value of $n_0$ indicates amplitudes pinned at the regularization scale.

\begin{figure}[tbp]
 \centering
 \includegraphics[width=0.94\textwidth]{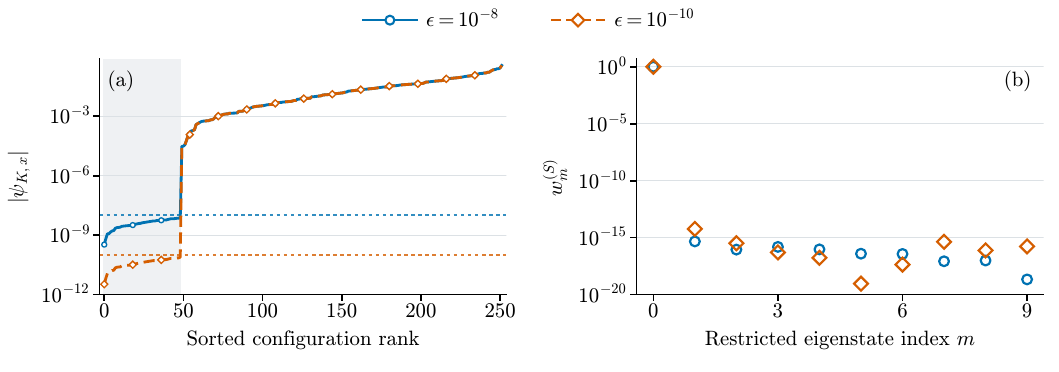}
 \caption{Reduced-support diagnostics after $K=1000$ updates for the first-excited-sign chain of Fig.~\ref{fig:numerics}, at $\eps=10^{-8}$ and $10^{-10}$. \textbf{(a)} Sorted terminal amplitudes. Horizontal dashed lines mark $\eps$; shaded ranks contain the $49$ configurations with $|\psi_{K,x}|\le\eps$. The surviving set $S=\{x:|\psi_{K,x}|>\eps\}$ contains the same $203$ configurations in both runs. \textbf{(b)} Weights $w_m^{(S)}=|\langle\varphi_m^{(S)},\widehat\psi_S\rangle|^2$ in the first ten eigenstates $\varphi_m^{(S)}$ of $H_{SS}$, ordered by increasing energy, where $\widehat\psi_S=\psi_{K,S}/\|\psi_{K,S}\|$.}
 \label{fig:HSS-supp}
\end{figure}

\section{Linearization of the self-consistent fixed-node map}
\label{sec:linearization-supp}
Let $\psi_\star$ be a normalized full-support fixed point with eigenvalue $E_\star$, assume the violating-edge pattern is constant nearby, and assume $E_\star$ is a simple ground-state eigenvalue of $F_\star=F(\psi_\star)$. The perturbations lie in the tangent space
\begin{equation}
 \mathcal T_{\psi_\star}=\{\eta:\langle\psi_\star,\eta\rangle=0\},
 \qquad
 P_\perp=I-|\psi_\star\rangle\langle\psi_\star|.
\end{equation}
Upon differentiating $F(\psi)\psi=H\psi$ at $\psi_\star$ we get:
\begin{equation}
 DF(\psi_\star)[\eta]\psi_\star
 =-(F_\star-H)\eta.
 \label{eq:derivative-identity}
\end{equation}

For a simple eigenvalue, first-order symmetric eigenvector perturbation theory gives
\begin{equation}
 DT_{\psi_\star}[\eta]
 =-(F_\star-E_\star)^+
 P_\perp DF(\psi_\star)[\eta]\psi_\star,
 \label{eq:pseudoinverse-derivative}
\end{equation}
where the pseudo-inverse acts on $\mathcal T_{\psi_\star}$. On that tangent space we can define:
\begin{equation}
 A=P_\perp(H-E_\star)P_\perp,
 \qquad
 B=P_\perp(F_\star-E_\star)P_\perp.
\end{equation}
Because $\psi_\star$ is a simple fixed-node ground state, $B\succ0$ and the restriction of $(F_\star-E_\star)^+$ is $B^{-1}$. Equations~\eqref{eq:derivative-identity} and \eqref{eq:pseudoinverse-derivative} yield
\begin{align}
 DT_{\psi_\star}
 &=B^{-1}P_\perp(F_\star-H)P_\perp\nonumber\\
 &=B^{-1}(B-A)
 =I-B^{-1}A.
 \label{eq:jacobian-supp}
\end{align}
Thus $\eta_{k+1}=DT_{\psi_\star}\eta_k+O(\|\eta_k\|^2)$. This is the one-vector self-consistent-field linearization of the eigenvector-dependent problem~\cite{cai2018eigenvector,bai2022sharp}.
\begin{theorem}[Ground-state local stability]
\label{thm:ground-stability-supp}
Assume the ground states of $H$ and $F(\phi_0)$ are simple and $\phi_0$ has full support. Then
\begin{equation}
 \rho(DT_{\phi_0})<1.
\end{equation}
\end{theorem}

\begin{proof}
Lemmas~\ref{lem:Fv-Hv} and \ref{lem:loewner} give
$F(\phi_0)\phi_0=E_0\phi_0$ and $F(\phi_0)\succeq H$, so $\phi_0$ is a ground state of $F(\phi_0)$ and hence a fixed point. By assumption it is simple.
On the tangent space at $\phi_0$,
$A=P_\perp(H-E_0)P_\perp\succ0$. Lemma~\ref{lem:loewner} gives
\begin{equation}
 B-A=P_\perp(F(\phi_0)-H)P_\perp\succeq0,
\end{equation}
so $0\prec A\preceq B$. The symmetric matrix
$C=B^{-1/2}AB^{-1/2}$ satisfies $0\prec C\preceq I$, and $B^{-1}A$ is similar to $C$. Every eigenvalue $\mu$ of $B^{-1}A$ lies in $(0,1]$, so the eigenvalues $1-\mu$ of Eq.~\eqref{eq:jacobian-supp} lie in $[0,1)$.
\end{proof}
\section{Morse index of excited-state fixed points}
\label{sec:morse-supp}

\begin{theorem}[Excited-state Morse index]
\label{thm:morse-supp}
Let $\phi_q$ be a full-support fixed point with nondegenerate energy $E_q$, and let $q$ be the number of lower eigenvalues of $H$, counted with multiplicity. Assume the violating-edge pattern is constant nearby and $\phi_q$ is a simple ground state of $F(\phi_q)$. Then $DT_{\phi_q}$ has exactly $q$ eigenvalues strictly larger than one; all remaining eigenvalues lie in $[0,1)$.
\end{theorem}

\begin{proof}
Let
\begin{equation}
 A_q=P_\perp(H-E_q)P_\perp,
 \qquad
 B_q=P_\perp(F(\phi_q)-E_q)P_\perp.
\end{equation}
The matrix $B_q$ is positive definite. The congruent matrix
$C_q=B_q^{-1/2}A_qB_q^{-1/2}$ has the same inertia as $A_q$, by Sylvester's law. The tangent-space eigenvalues of $A_q$ are $E_j-E_q$, $j\ne q$, so exactly $q$ are negative. Moreover, Lemma~\ref{lem:loewner} gives $C_q\preceq I$. Each negative eigenvalue $\mu$ of $C_q$ produces a Jacobian eigenvalue $1-\mu>1$, whereas each positive eigenvalue lies in $(0,1]$ and produces a Jacobian eigenvalue in $[0,1)$.
\end{proof}

If the fixed-node construction disconnects the configuration graph, the fixed-node ground space can be degenerate. Then $B_q$ is singular and Eq.~\eqref{eq:jacobian-supp} cannot be used without separating directions inside and transverse to that ground space. Theorems~\ref{thm:ground-stability-supp} and \ref{thm:morse-supp} make no claim about this case.

\section{Additional numerical results}

We describe an exactly solvable benchmark for amplitude repair, and present additional numerical tests of support collapse and convergence.

\subsection{Exactly solvable non-stoquastic ring}
\label{sec:exact-solve}
Consider an even periodic ring with
\begin{equation}
 H_n=b\sum_{i=1}^n(I-X_i)
     -a\sum_{i=1}^n(I-X_iX_{i+1}),
 \qquad 0<a<\frac b2.
 \label{eq:solvable-ring-H}
\end{equation}
In the computational basis, a single-spin flip has matrix element $-b$, whereas an adjacent double flip has matrix element $+a$.  The product of the edge signs around the triangle formed by two consecutive single flips and the corresponding double flip is positive, so no diagonal $\{\pm1\}$ gauge can make every off-diagonal element nonpositive.  Nevertheless, the unique ground state is completely positive,
\begin{equation}
 \phi_0=\ket{+}^{\otimes n}
 =2^{-n/2}\sum_x\ket{x},
 \qquad E_0=0,
 \qquad \Delta_H=2(b-2a).
 \label{eq:solvable-ring-ground}
\end{equation}
Indeed, in the common $X$-eigenbasis a configuration with $N_-$ minus spins and $N_{\rm dw}$ domain walls has energy $2bN_--2aN_{\rm dw}\geq2(b-2a)N_-$.  Thus positivity of the final ground-state amplitudes does not remove the computational-basis projector sign problem: the short-time propagator assigns negative weight $-a\,\delta\tau$ to a double flip. To introduce a pure amplitude error with the exact signs, let $S=\{1,3,\ldots,n-1\}$,
\begin{equation}
 \chi(x)=(-1)^{\sum_{i\in S}x_i},
 \qquad
 v_\theta(x)=\frac{e^{\theta\chi(x)}}{\sqrt{2^n\cosh(2\theta)}}.
 \label{eq:solvable-ring-guide}
\end{equation}
The exact state is $v_0$, while $\theta\neq0$ biases the relative amplitudes of the two $\chi=\pm1$ sectors without changing any sign.  Fixed node deletes all $+a$ double-flip edges, retains the $-b$ single-flip edges, and adds the diagonal compensation
\begin{equation}
 F_\theta(x,x)=n(b-a)+an e^{-2\theta\chi(x)}.
 \label{eq:solvable-ring-F}
\end{equation}
Because $F_\theta$ depends on a configuration only through $\chi(x)$, its positive ground state lies in the same one-parameter family.  One exact iteration is therefore closed as
\begin{equation}
 T(v_\theta)=v_{g(\theta)},
 \qquad
 g(\theta)=\frac12\operatorname{arsinh}
 \!\left[\frac{2a}{b}\sinh(2\theta)\right].
 \label{eq:solvable-ring-map}
\end{equation}
If $c=2a/b<1$, we have
\begin{equation}
 |g'(\theta)|=
 \frac{c\cosh(2\theta)}{\sqrt{1+c^2\sinh^2(2\theta)}}<1,
 \qquad g'(0)=\frac{2a}{b}.
 \label{eq:solvable-ring-contraction}
\end{equation}
Hence $\theta=0$ is the unique fixed point and $\theta_k\to0$ geometrically under $\theta_{k+1}=g(\theta_k)$.  This gives an analytic example in which the original Hamiltonian is sign-frustrated in the sampling basis, yet exact fixed-node self-consistency reconstructs the uniform ground-state amplitudes from their positive sign pattern. Fig.~\ref{fig:solvable-ring-iteration} shows the map and its derivative for three coupling ratios, together with the monotonic decrease of $|\theta_k|$ from randomly chosen initial values within this guiding wavefunction family.

\begin{figure}[tbp]
 \centering
 \includegraphics[width=\textwidth]{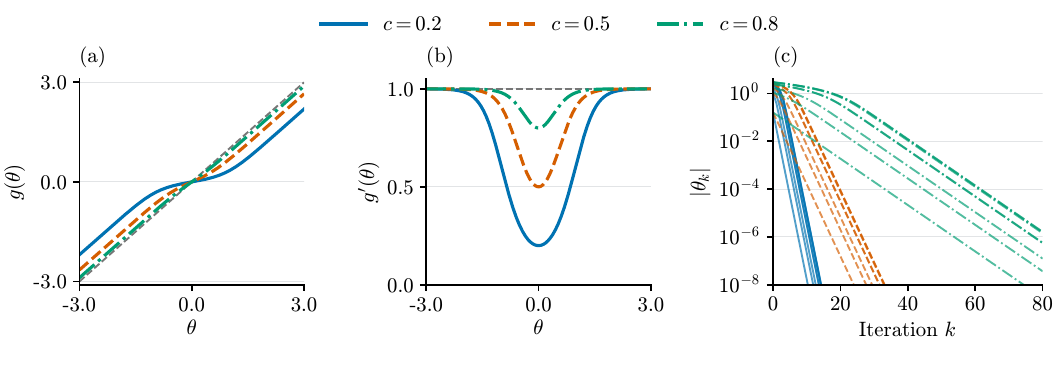}
 \caption{Exact amplitude iteration for the non-stoquastic ring, with $c=2a/b=0.2,0.5,0.8$. \textbf{(a)} The update $g(\theta)$ in Eq.~\eqref{eq:solvable-ring-map}; the dashed gray line is $g(\theta)=\theta$. \textbf{(b)} The derivative in Eq.~\eqref{eq:solvable-ring-contraction}, which remains below the dashed gray reference $g'=1$ for every finite $\theta$. \textbf{(c)} Iteration histories for eight initial values drawn uniformly from $[-3,3]$, with the same values used for each $c$. Each trajectory preserves its sign while $|\theta_k|$ decreases monotonically. The logarithmic vertical scale resolves the asymptotic geometric decay, $|\theta_{k+1}|/|\theta_k|\to c$. Colors and line styles denote the same coupling ratios in all panels.}
 \label{fig:solvable-ring-iteration}

\end{figure}

\begin{figure}[t]
  \centering
  \begin{subfigure}[t]{0.28\textwidth}
    \caption{}\label{fig:triangular-ground-energy}
    \includegraphics[width=\linewidth]{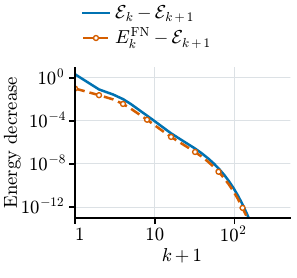}
  \end{subfigure}\hfill
  \begin{subfigure}[t]{0.21\textwidth}
    \caption{}\label{fig:triangular-ground-weights}
    \includegraphics[width=\linewidth]{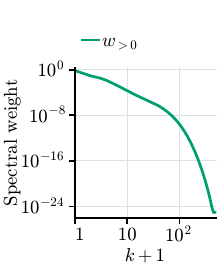}
  \end{subfigure}\hfill
  \begin{subfigure}[t]{0.28\textwidth}
    \caption{}\label{fig:triangular-excited-energy}
    \includegraphics[width=\linewidth]{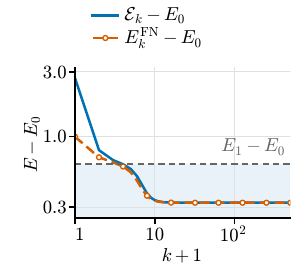}
  \end{subfigure}\hfill
  \begin{subfigure}[t]{0.21\textwidth}
    \caption{}\label{fig:triangular-excited-weights}
    \includegraphics[width=\linewidth]{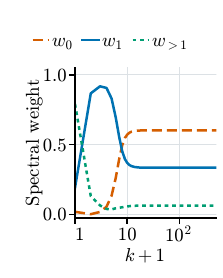}
  \end{subfigure}
  \caption{Fixed-node iteration on an open $4\times2$ triangular strip: \textbf{(a,b)} ground-state signs; \textbf{(c,d)} first-excited-state signs. Panel \textbf{(a)} shows $\mathcal{E}_k-\mathcal{E}_{k+1}$ and $E_k^{\rm FN}-\mathcal{E}_{k+1}$; panel \textbf{(c)} shows $\mathcal{E}_k-E_0$ and $E_k^{\rm FN}-E_0$, with the dashed reference $E_1-E_0$ and shading below it. The physical energy crosses below $E_1$ at $k=3$ and finishes approximately $0.306$ below it. Panels \textbf{(b,d)} show spectral weights, using the colors and notation of Fig.~\ref{fig:numerics}. The decay of $w_{>0}$ in \textbf{(b)} shows $\psi_k \rightarrow \phi_0$, while growing $w_0$ in \textbf{(d)} reveals the lower-energy eigenstates gain finite weight. Horizontal axes, energy axes, and the weight axis in \textbf{(b)} are logarithmic; the weight axis in \textbf{(d)} is linear. Energy differences at or below $10^{-13}$ are omitted in \textbf{(a)}. Both runs use the square-root guide with $\eps=10^{-8}$.}
  \label{fig:triangular-supp}
\end{figure}

\begin{figure}[tbp]
 \centering
 \begin{minipage}[t]{0.49\textwidth}
  \centering
  \textbf{(a)}\\[-2pt]
  \includegraphics[width=\linewidth]{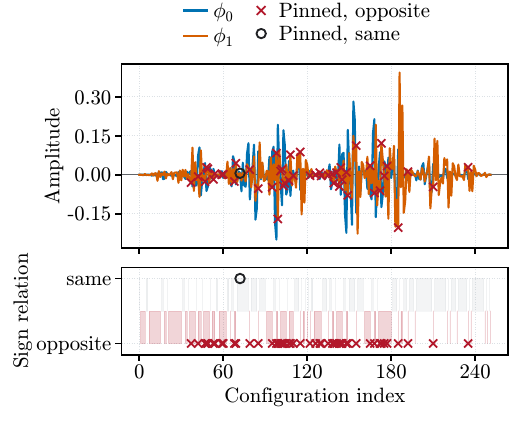}
 \end{minipage}
 \hfill
 \begin{minipage}[t]{0.49\textwidth}
  \centering
  \textbf{(b)}\\[-2pt]
  \includegraphics[width=\linewidth]{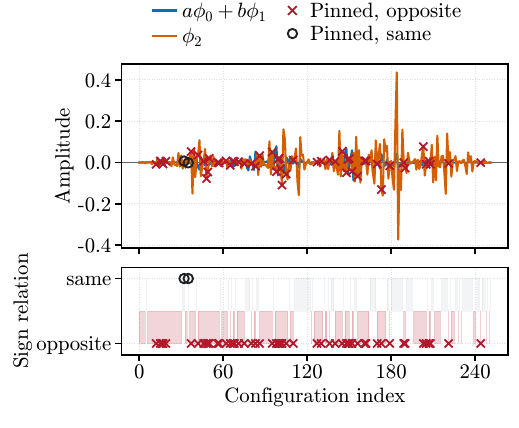}
 \end{minipage}

 \caption{Sign alignment for the disordered $L=10$ $J_1$-$J_2$ chain with $J_1=\Delta=1$, $J_2=0.78$, and $h_{\rm dis}=0.1$. Upper panels show amplitudes versus configuration index: \textbf{(a)} the ground state $\phi_0$ and first excited state $\phi_1$; \textbf{(b)} the second excited state $\phi_2$ and the lower-state combination $\chi=a\phi_0+b\phi_1$, with $a\simeq-0.31$ and $b\simeq0.01$ fitted by least squares to $-\phi_2$ on the pinned configurations. Red crosses and open black circles mark configurations with $|\psi_{1000,x}|\le\eps=10^{-10}$ where the two amplitudes have opposite and equal signs, respectively. The lower strips show the sign relation for all $252$ configurations and repeat these markers for the pinned set.}
 \label{fig:alignment-supp}
\end{figure}

\subsection{Numerics for 2D Hamiltonian}

Fig.\ref{fig:triangular-supp} provides a complementary test on a two-dimensional frustrated geometry: an open $4\times2$ triangular strip with eight spins, restricted to the $S^z_{\rm tot}=0$ sector of dimension $70$.  We take $J=\Delta=1$, add a weak longitudinal random field with $h_{\rm dis}=0.12$ to remove accidental degeneracies, and initialize random amplitudes with the exact signs of the chosen target; both runs use $\eps=10^{-8}$.  In panels (a,b), the ground-state-sign trajectory converges to $E_0=-3.704\ldots$ and its ground-state weight approaches unity.  In panels (c,d), the first-excited-state weight initially increases but the unstable lower-energy direction subsequently develops: after $500$ iterations the ground-, target-, and higher-state weights are approximately $0.603$, $0.334$, and $0.063$, respectively, while the energy $-3.38\ldots$ lies below $E_1=-3.075\ldots$.  Fifteen of the $70$ amplitudes are pinned at the regularization scale, consistently with support collapse.  Thus the chamber-dependent behavior seen in one dimension persists on this frustrated two-dimensional cluster, although this single size and disorder realization is not a finite-size-scaling claim. Figure~\ref{fig:alignment-supp} examines whether the amplitudes driven to the regularization scale are associated with sign conflicts with lower-energy states.  In the first-excited-state run, $48$ of the $49$ pinned configurations have opposite signs in $\phi_1$ and $\phi_0$.  In the second-excited-state run, a post hoc least-squares vector $\chi\in\operatorname{span}\{\phi_0,\phi_1\}$ has the opposite sign to $\phi_2$ on $53$ of the $55$ pinned configurations.  These finite-system diagnostics are consistent with sign-selective support collapse rather than indiscriminate amplitude suppression; they do not by themselves establish an exact reduced-support fixed point or a general convergence mechanism.

\section{Sign and amplitude reconstruction}
\label{sec:glassysign}
For $n$ qubits the real configuration space has dimension $2^n$. A full-support state therefore has $2^{2^n}$ componentwise sign assignments, or $2^{2^n-1}$ after identifying one global sign. Westerhout \textit{et al.} provide an amplitude-to-sign map that is complementary to our sign-to-amplitude dynamics~\cite{westerhout2023many}. Given normalized amplitudes $a_x\geq0$ on a connected $K$-configuration subgraph $\mathcal C$, they associate the sign variables $s_x\in\{\pm1\}$ with the auxiliary Ising objective
\begin{equation}
 \begin{aligned}
 \mathcal H_{\rm sgn}(s;a)
   &=\sum_{\{x,y\}\in E_{\mathcal C}}H_{xy}a_xa_y s_xs_y,\\
 \mathcal G_H(a)&\approx\argmin_s\mathcal H_{\rm sgn}(s;a).
 \end{aligned}
 \label{eq:sign-ising-map}
\end{equation}
Their edge-ordered greedy map $\mathcal G_H$ recovers exact or high-overlap signs, even with slightly-noisy input amplitudes, in $O(K\log K)$ time. Here $K$ is the size of the auxillary Ising graph optimizied by the greedy algorithm; keeping $K$ polynomial requires sampling a configuration cluster. For $n$-spins their second-order-extension variational monte carlo proposal has an estimated sign-update cost $O(n^4\log n)$. We propose combining their sign update with $m$ fixed-node iterations on the same configuration set to update the amplitudes:
\begin{equation}
 \begin{aligned}
 s^{(r+1)}&=\mathcal G_H(a^{(r)}),\\
 \psi^{(r,0)}&=s^{(r+1)}\odot a^{(r)},\\
 a^{(r+1)}&=\left|T^m\!\left(\psi^{(r,0)}\right)\right|.
 \end{aligned}
 \label{eq:sign-amplitude-bootstrap}
\end{equation}
where $\odot$ denotes componentwise multiplication and $m$ is the number of fixed-node updates between sign reconstructions. If the sign step returns $s^{(r+1)}=\sgn(\phi_0)$, our convergence result makes the amplitude update almost exact for large $m$, under the hypotheses stated above. Fig.\ref{fig:reconstruction-supp} shows the amplitude error after $m$ iterations when some initial signs are incorrect. In the overlap errors $e_{\mathcal A}=1-O^{\mathcal A}$ and $e_{\mathcal S}=1-O^{\mathcal S}$, the robustness test of Ref.~\cite{westerhout2023many} probes an amplitude-to-sign response $e_{\mathcal S}\lesssim f(e_{\mathcal A})$, whereas Fig.~\ref{fig:reconstruction-supp} probes the reverse response $e'_{\mathcal A}\approx g(e_{\mathcal S})$. A direct test of a self-correcting regime would place both maps on the same model and ensemble and verify $g(f(e))<e$ over a nonzero interval. Neither work yet proves this contraction or the efficiency of the composed stochastic, sampled iteration. In particular, the $O(K\log K)$ bound applies only to the greedy sign subroutine.

\begin{figure}[tbp]
 \centering
 \includegraphics[width=0.8\textwidth]{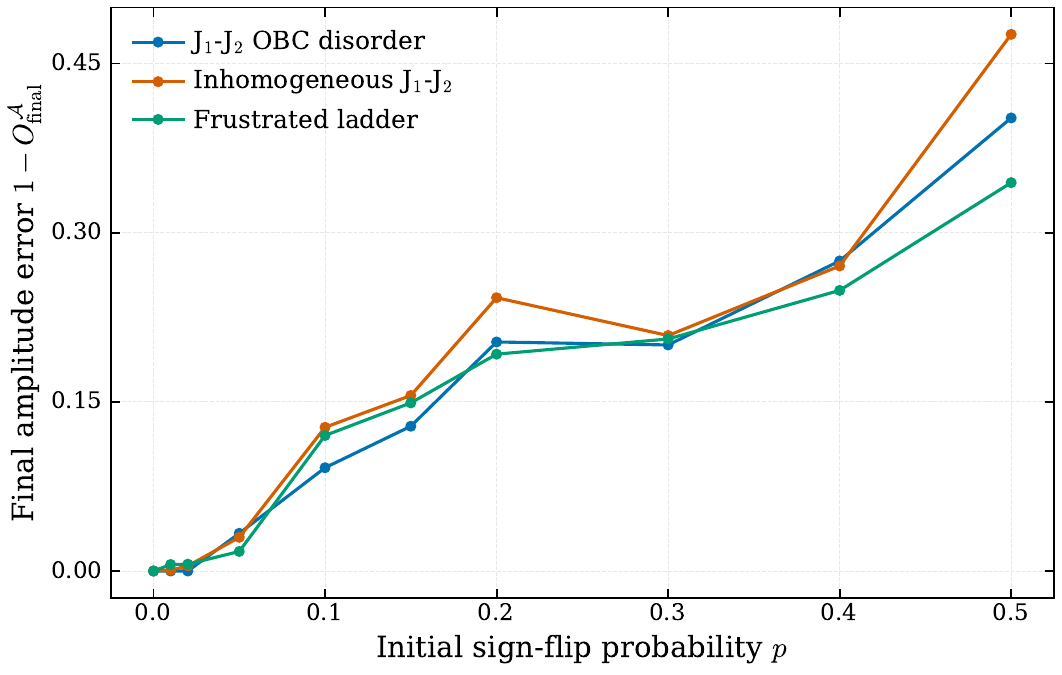}
 \caption{Finite-step sign-to-amplitude robustness for three $L=10$ ground-state problems. Starting from a fixed random positive amplitude vector, each exact ground-state sign is independently reversed with probability $p$, and the state is evolved for $m=500$ deterministic fixed-node iterations. For the weighted sign overlap $O^{\mathcal S}_{\rm in}=\sum_x|\phi_{0,x}|^2s_x^{\rm exact}s_x^{\rm in}$, this protocol gives $\mathbb E[1-O^{\mathcal S}_{\rm in}]=2p$. The plotted error is $1-O^{\mathcal A}_{\rm final}$, where $O^{\mathcal A}_{\rm final}=\sum_x|\phi_{0,x}|\,|\psi_{m,x}|$ for normalized states. At $p=0$ the error reaches the numerical convergence floor, as predicted by exact amplitude reconstruction in the correct sign chamber; increasing $p$ probes imperfect output from a putative sign-reconstruction step. Curves are single seeded scans (not ensemble averages) for the disordered open $J_1$-$J_2$, inhomogeneous $J_1$-$J_2$, and frustrated-ladder models, so no general contraction claim for the alternating scheme is inferred.}
 \label{fig:reconstruction-supp}
\end{figure}
\end{document}